\documentclass[11pt]{article}

\usepackage{hyperref}
\usepackage{gnuplottex}
\usepackage{fullpage}
\usepackage{amsmath}
\usepackage{amssymb}

\newenvironment{proof}{\noindent {\bf Proof.  }}{\hfill$\Box$}
\newtheorem{theorem}{Theorem}[section]

\newtheorem{conjecture}[theorem]{Conjecture}
\newtheorem{proposition}[theorem]{Proposition}

\newtheorem{definition}[theorem]{Definition}
\def \poly { \text{\rm poly} }
\def \DTS {{\sf DTS}}
\def \R {{\mathbb R}}

\def \T {{\sf TIME}}
\def \Sig[#1] {{\sf \Sigma}_{#1} }
\def \Pie[#1] {{\sf \Pi}_{#1} }

\def \NT {{\sf NTIME}}
\def \NTIME {\NT}

\def \coNT {{\sf coNTIME}}

\def \coNTIME {\coNT}

\def \isin {\subseteq}

\def \eps {\varepsilon}

\title{Applying Practice to Theory}
\author{Ryan Williams\footnote{School of Mathematics, Institute for Advanced Study,
Princeton, NJ 08540, USA. Email: {\tt ryanw@math.ias.edu}. This material is based on work supported by NSF grants CCF-0832797 and DMS-0835373. An alternative version of this article will appear in SIGACT News.}}\date{}
\begin{document}
%\sloppy

% If a maketitle is called, the default is to not use the
% fancy page style which leads to no ACM SIGACT News footline.
% So an alternative maketitle has been provided.
% Note added 2006/9/26: but due to the switch to footers being
% provided centrally by the head editor, via a \thispagestyle{empty}
% this is overridden right after the \SIGACTmaketitle, below.
\maketitle
%\thispagestyle{empty}
%\vspace*{-0.7in}

% your contributions go here.

\begin{abstract} How can complexity theory and algorithms benefit from practical advances in computing? We give a short overview of some prior work using practical computing to attack problems in computational complexity and algorithms, informally describe how linear program solvers may be used to help prove new lower bounds for satisfiability, and suggest a research program for developing new understanding in circuit complexity.
\end{abstract}

\section{Introduction}

As hardware becomes steadily more powerful, computer scientists should repeatedly ask themselves {\em what can we do with all the spare computing power at our disposal?} There have been many inventive and exciting answers to this question in the form of distributed computing projects, from the study of protein folding~(Folding@Home) to the improvement of climate models~(Climateprediction.net) to the search for extraterrestrial intelligence~(SETI@Home) to the mathematically intriguing~(the Great Internet Mersenne Prime Search).

I would like to suggest that some spare computing power should be devoted towards achieving a better understanding of computation itself: solving significant open problems in theoretical computer science. I would hope that this suggestion is indisputable. Considering all that we do not understand in the theory of computing, exploiting extraneous computation for improving basic knowledge should be a priority. An attractive property of this suggestion is its potential for self-improvement. Ideally, improvements in basic knowledge can lead to more spare cycles in the future, leading to further advances in basic knowledge, and so on.

While the above looks splendid in writing, of course it is not clear how to better the causes of theory in this way. In fact this problem may be exceptionally difficult in most cases of interest. My hope in this survey to encourage readers to think more seriously about the problem. %I am not convinced that theoreticians spend much time thinking about how to use computers extensively in their research. On the contrary, I have encountered attitudes suggesting that the use of computers to help prove theorems is simply defeatist. Personally I would rather see a problem be solved any way at all, than to see years pass with theory beating its collective head against the wall, stuck in a local maximum of thought. It is always possible that a computer-assisted proof of an important theorem could be simplified and better understood with time (and this has happened). Truth first; ask questions later.
Proofs reliant on computer calculation have become increasingly more common in mathematics. Examples include the proof of the four color theorem by Appel and Haken~\cite{AH1,AH2} which has been greatly simplified~\cite{RSST}, Hales' proof of the Kepler Conjecture~\cite{Hales}, Hass and Schlafly's proof of the double bubble conjecture~\cite{HS}, Lam's proof that there is no finite projective plane of order 10~\cite{Lam}, and McCune's proof of the Robbins Conjecture~\cite{Robbins}. In general, computers can play a much greater role than simply discovering or verifying proofs, the end products of mathematical research. To prove theorems we need to formalize a proof system, and in doing so we may have to finagle messy details. Computers are possibly more useful in helping us discover new computational artifacts suggestive of a deeper paradigm, such as an ingeniously tiny circuit for $7\times 7$ matrix multiplication which inspires us to look for something better.

At this point I should note a distinction between two usages of computers in proofs. For convenience, I will designate them {\em infeasibly checkable} and {\em feasibly checkable}. In the former, a computer is needed to generate the proof of the theorem, and the ``proof'' is the code of some program, the trace of that program's execution, and our observation that the program output the appropriate answers. I personally have no problem with such proofs so long as they are carefully reviewed, but acknowledge that they stretch the bounds of what one typically calls a proof. In contrast, feasibly checkable proofs may require strenuous computing effort, but once obtained they can be checked by mathematicians within a reasonable time frame. Restricting ourselves to feasibly checkable proofs keeps the task of proof-finding roughly within {\sf NP}, and I would like to strongly promote this usage of computers whenever possible.

First I will survey a few ideas in this vein that have been introduced in algorithms and complexity theory. Then I will discuss some recent work using computers to find proofs of time lower bounds in restricted models and small circuits. I will not give prescriptions for solving your major open problems via distributed computing with spare desktop cycles. But I do wish that this article helps you consider the possibility.

\section{Some Prior Work}

Let me begin by saying that I cannot hope to cover all the innovative uses of computers in algorithms and complexity. My goal is to merely point out a few theory topics that I know where practical computing has made a noteworthy impact.

\subsection{Moderately Exponential Algorithms}

%The design and analysis of moderately exponential algorithms for hard problems has seen a huge revival in the last decade.
In the area of moderately exponential algorithms, the goal is to develop faster algorithms for solving {\sf NP}-hard problems exactly. Of course we do not expect these faster algorithms to be remotely close to polynomial time. Instead we settle for exponential algorithms that are still significant improvements over exhaustive search. For example, a $O(1.9^n \cdot \poly(m))$ time algorithm for Boolean Circuit Satisfiability on $n$ variables and $m$ gates would be very interesting, but not necessarily for any practical reason. The ability to avoid brute-force search in such a general case would be an amazing discovery in itself. Currently we have no idea how to construct such an algorithm, or even if its existence would imply something unexpected. However, we can solve the 3-SAT problem in $O(1.33^n)$ time with a randomized algorithm, by taking advantage of the structure of $3$-CNF. This is due to Iwama and Tamaki~\cite{IT} and is just the latest of a long line of results on the problem.

\subsubsection{Analyzing Exponential Algorithms}

Typically, one of the central problems in an exponential algorithms paper is to prove good upper bounds on the running time of some algorithm which solves some hard problem. (Usually the correctness of the algorithm is straightforward.) Much effort has been undertaken to better understand the behavior of backtracking algorithms. To aid the discussion, let us work with a toy example. For a node $v$ in a graph, let $N(v)$ denote the set of $v$'s neighbors. Consider the following algorithm for solving Minimum Vertex Cover:

\medskip

{\narrower

{\em If all nodes have degree at most two, solve in polytime. If a node has degree $0$, remove it from the graph and recurse on the remaining graph. If a node has degree $1$, then remove it and its neighbor $u$, recurse on the resulting graph getting a cover $C$, and return $C \cup \{u\}$. Otherwise take a node $v$ of highest degree. Recurse on the graph with $v$ removed, getting a cover $C_1$. Recurse on the graph with $v$ and $N(v)$ removed, getting a cover $C_2$. Return the minimum of $C_1 \cup \{v\}$ and $C_2 \cup N(v)$.}

}

\medskip

We won't explain how to solve the degree-two case here, but leave it to the reader. How do we analyze the runtime of such an algorithm? A natural approach is to try writing a recurrence. Let $T(n)$ be the runtime on an $n$ node graph. In the two recursive calls of the algorithm, we remove at least one node and at least four nodes, respectively. (We remove at least four nodes because $|N(v)|\geq 3$.) This gives us the recurrence \[T(n) \leq T(n-1) + T(n-4) + \poly(n).\] Solving this recurrence in the usual way (by finding a real root of $f(x) = 1-1/x-1/x^4$) we find that $T(n) \leq 1.39^n$. While exponential, this is still better than the obvious $2^n$ algorithm, since we can now handle instances of more than double the size in the same running time.

Our analysis certainly has much slack. One can imagine cases where the number of nodes removed increases by much more, and we have not taken them into account. More generally, it is not clear that $n$ is the best ``progress measure'' for the algorithm. Perhaps if we count the number of edges in the graph instead, we may find a better runtime bound for sparse graphs. Indeed, in the first recursive call, at least three edges are removed (since we chose a $v$ of largest degree) and in the second call we remove at least five edges: at least three edges for the neighbors of $v$, and at least two additional edges since the neighbors have degree at least two. We have \[T(m) \leq T(m-3) + T(m-5) + O(\poly(n)),\] leading to $T(m) \leq O(1.19^m)$. For sufficiently sparse graphs, this improves $1.39^n$. More ambitiously, we could try to capture both observations with the double recurrence \[T(m,n) \leq T(m - 3,n-1) + T(m - 5,n-4) + O(\poly(n)).\] Now how do we deal with this? One way is to convert the double recurrence into a single one, by letting $k = \alpha_1 m + \alpha_2 n$. Then \[T(k) \leq T(k - 3\alpha_1 - \alpha_2) + T(k - 5\alpha_1 - 4\alpha_2) + O(\poly(k)).\] So $T(k) \leq O(c^k)$ where $c$ is a real solution to $1-1/x^{3\alpha_1 + \alpha_2}-1/x^{5\alpha_1 + 3\alpha_2} = 0$. For an example when $\alpha_1=1/2$ and $\alpha_2=1$ then the runtime bound is $O(1.21^{m/2+n})$. In this way, we can interpolate between the two time bounds.

In general, optimizing this sort of analysis can become terribly complicated. When there are many possible cases in the algorithm, and different variable measures are decreasing at different rates, the analysis becomes intractable to carry out by hand. However, researchers have found ways to apply computers to the problem. Eppstein~\cite{Eppstein} showed that multivariate recurrences similar to the above can be approximately solved efficiently, by expressing the problem as a {\em quasi-convex program}. This has become a very useful tool. For example, we could keep track of the number of nodes $n_i$ of degree $i$ in its time recurrence expressions, for $i \geq 2$ (and for sufficiently large $k$, we lump the number of degree $\geq k$ nodes into a single quantity $n_{\geq k}$). This can also be converted into a single variable recurrence, introducing $\alpha_i$ weights for each $n_i$.

For some algorithms we can get surprisingly good time bounds in terms of the total number of nodes: quasiconvex optimization uncovers interesting $\alpha_i$'s. Intuitively, this makes sense, because removing high degree nodes should contribute more to the progress of the algorithm than removing those of low degree. Fomin, Grandoni, and Kratsch~\cite{FGK1,FGK2} have found $O(1.52^n)$ and $O(1.23^n)$ algorithms for Dominating Set and Maximum Independent Set/Minimum Vertex Cover respectively, by performing analyses of the above kind on simple new algorithms, using Eppstein's computer approach to determine optimal settings of $\alpha_i$. For more details, see the survey~\cite{FGK3}. Scott and Sorkin~\cite{SS} have analyzed algorithms for Max 2-CSP and related graph problems with a somewhat similar approach, keeping track of the degrees of neighbors in the recurrence as well. Their approach formulates the analysis with linear programming instead. Provided the original recurrences are reasonably sized, the above approaches can generate feasibly checkable proofs: after the optimization has found appropriate weights, one can often manually check that the recurrence works out.

\subsubsection{Case Analysis of Exponential Algorithms}

Another approach is to have a computer check that a recursive algorithm will admit an efficient time recurrence, over all the possible inputs up to a certain size. Since many recursive backtracking algorithms work very locally (they only look at a subgraph of finite size around a specially chosen node) this sort of case analysis is sometimes enough to ensure a good upper bound on the running time. However, this style of approach typically does not lead to feasibly checkable proofs of upper bounds, and the case analysis is done by computer.

Robson~\cite{Robson}, in an unpublished technical report, has written a program to do this for a Maximum Independent Set algorithm, proving that the algorithm runs in $O(2^{n/4})$ time. In particular, a lengthy and very complex extension of the above toy algorithm is presented and analyzed case-by-case, using a computer to enumerate many of the possible cases.

In the toy algorithm, there were some special cases prior to backtracking: removing nodes of degree less than two, and solving instances with only degree two nodes. Rules similar to degree $0$ and degree $1$ node removal are called {\em simplification rules}. In general, these are short polytime rules that allow one to reduce the size of an instance practically for free, provided that a certain substructure exists. Fedin, Kojevnikov, and Kulikov~\cite{FK06,KK06} developed a natural formalism for expressing special cases in SAT problems, which made it not only possible for a computer to perform case analyses, but also search for new simplification rules for MAX SAT and SAT on its own, resulting in faster new algorithms. For earlier work of this kind, see~\cite{NS,Gramm}.

Could it be possible to search over (or reason about) {\em all} recursive backtracking algorithms in some sense, and show exponential limitations on solving problems like SAT? Here I am using ``all'' very loosely; without clever simplification rules, exponential lower bounds on treelike resolution already give lower bounds on simple backtracking. Perhaps there is a $2^{\eps n}$ algorithm for {\em every} $\eps > 0$, by using a sufficiently complicated backtracking algorithm. Many seem to disbelieve in this possibility; some work has articulated this belief, in some sense. For instance, Alekhnovich et al.~\cite{Borodin} formalized a model for backtracking algorithms, proving that 3SAT requires $2^{\Omega(n)}$ time in their model. However their proof uses a SAT instance that encodes a linear system of equations over $GF(2)$, which can be solved trivially in polynomial time. Such results teach us that we should not be myopically focused on one specific algorithmic technique.

\subsection{Approximability and Inapproximability}

Since the early 90's there has been significant progress in the study of hard-to-approximate problems, aided by the celebrated PCP theorem~\cite{AS,ALMSS}. An algorithm is a $\rho$-approximation for a minimization problem $\Pi$ if on all inputs the algorithm outputs a solution that has value at most $\rho$ times the minimum value of any solution. In the maximization case, the output solution must have cost at least $\rho$ times the maximum. Note when $\Pi$ is a minimization (maximization) problem, we have $\rho \geq 1$ ($\rho \leq 1$), respectively. A prime objective in the study of approximability is to determine for which $\rho$ a problem can be $\rho$-approximated in polytime, and for which $\rho$ a problem is ${\sf NP}$-hard to $\rho$-approximate. Several surprisingly tight results are known; for instance, a random assignment satisfies at least $7/8$ of the clauses in any 3-CNF formula (with three distinct variables in each clause), yet H\aa stad showed~\cite{Hastad} that it is {\sf NP}-hard to satisfy a number of clauses that is at least $7/8+\eps$ of the optimum, for any $\eps > 0$. That is, a polytime $(7/8+\eps)$-approximation would imply ${\sf P} = {\sf NP}$.

Here I will briefly survey a couple of works in the study of approximation that rely heavily on computer power to achieve their results. %In the first, Trevisan, Sorkin, Sudan, and Williamson~\cite{TSSW} prove inapproximability results for problems, by finding new reductions from hard-to-approximate problems. In the second, Zwick~\cite{Zwick} uses a computer system designed with interval arithmetic to prove that certain algorithms produce optimal approximation ratios.

\subsubsection{Gadgets via Computer}

Everyone who has seen an {\sf NP}-completeness reduction knows what a gadget is. Treating one problem $\Pi$ as a programming language, you try to express pieces of an instance of $\Pi'$ by constructing gadgets, simple components that can be used over and over to express instances of $\Pi'$ as instances of $\Pi$.

To illustrate, consider the standard reduction from 3-SAT to MAX 2-SAT due to Garey, Johnson, and Stockmeyer~\cite{GJS}. One can transform any 3-CNF formula $F$ into a 2-CNF formula $F'$, by replacing each clause of $F$ such as $c_i = (\ell_1 \vee \ell_2 \vee \ell_3)$, where $\ell_1$, $ \ell_2$, and $\ell_3$ are literals, with the ``gadget'' of 2-CNF clauses \[ (\ell_1), (\ell_2), (\ell_3), (y_i), (\neg \ell_1 \vee \neg \ell_2), (\neg \ell_2 \vee \neg \ell_3), (\neg \ell_1 \vee \neg \ell_3), (\ell_1 \vee \neg y_i), (\ell_2 \vee \neg y_i), (\ell_3 \vee \neg y_i),\] where $y_i$ is a new variable. If an assignment satisfies $c_i$, then exactly $7$ of the $10$ clauses in the gadget can be satisfied by setting $y_i$ appropriately. If an assignment does not satisfy $c_i$, then exactly $6$ of the $10$ can be satisfied. Therefore $F$ is satisfiable if and only if $7/10$ of the clauses in $F'$ can be satisfied. This reduction also says something about the approximability of MAX 3-SAT: \begin{proposition} If an algorithm is a $(1-\eps)$-approximation for MAX 2-SAT, then by applying the reduction one can obtain a $(1-7\eps)$-approximation to MAX 3-SAT.\footnote{To see this, let $m_3$ be the number of clauses in the original 3-CNF $F$, let $m_3^* \leq m_3$ be the optimal number of clauses that can be satisfied in $F$, and let $\widehat{m_3}$ be the number of clauses in $F$ satisfied by running the MAX 2-SAT approximation on $F'$ and translating the output back to an assignment on the variables of $F$. By our assumption we have $(1-\eps)\leq \frac{6m_3+\widehat{m_3}}{6m_3 + m_3^*}$. By algebraic manipulation and the fact that $m_3^* \leq m_3$ we derive that $\frac{\widehat{m}_3}{m^*_3} \geq 1-7\eps$.}\end{proposition} Recall we mentioned that MAX 3-SAT does not have a polytime $(7/8+\eps)$-approximation unless ${\sf P} = {\sf NP}$. Hence the proposition implies that MAX 2-SAT cannot be $(55/56+\eps)$-approximated in polytime unless ${\sf P} = {\sf NP}$. So gadgets can be used to extend inapproximability results from one problem to another.\footnote{Note that the best known inapproximability result for MAX 2-SAT uses a different gadget reduction, cf.~\cite{Hastad}.}

How good is the above gadget from 3-SAT to 2-SAT? Could we find a gadget that implies stronger inapproximability for MAX 2-SAT? To address these kinds of questions, Trevisan, Sorkin, Sudan, and Williamson~\cite{TSSW} formalized gadgets, following~\cite{BGS}:

\begin{definition} Let $\alpha, \ell, n \geq 1$,  let $f:\{0,1\}^k \rightarrow \{0,1\}$, and let ${\cal F}$ be a family of functions from $\{0,1\}^{k+n}$ to $\{0,1\}$. An $\alpha$-gadget reducing $f$ to ${\cal F}$ is given by a set of auxiliary variables $y_1,\ldots,y_n$ and weights $w_j \geq 0$ coupled with constraints $C_j \in {\cal F}$, where $j = 1,\ldots,\ell$. For every $a \in \{0,1\}^k$, \begin{itemize} \item If $f(a) = 1$ then $(\forall b \in \{0,1\}^n) \sum_j w_j C_j(a,b) \leq \alpha$ and $(\exists b \in \{0,1\}^n) \sum_j w_j C_j(a,b) = \alpha$. \item If $f(a) = 0$ then $(\forall b \in \{0,1\}^n) \sum_j w_j C_j(a,b) \leq \alpha - 1$.\end{itemize}\end{definition}

Note it is fine to place weights on constraints: to ``unweight'' them, we can simply make a number of copies of each constraint in the instance, proportional to the weights. Observe the reduction from 3-SAT to MAX 2-SAT is a $7$-gadget reducing $f(x_1,x_2,x_3) = x_1 \vee x_2 \vee x_3$ to the family of functions representable by 2-variable clauses, where $n=1$, $\ell=10$, and $w_j = 1$ for all $j$.

First,~\cite{TSSW} showed that if we fix the number of auxiliary variables $n$, then the requirements in the gadget definition can be described by a large number ($|{\cal F}|^{\ell}$) of linear programs, with a large number of inequalities in each linear program. That is, by specifying $n$ and a tuple $(C_1,\ldots,C_{\ell}) \in {\cal F}^{\ell}$, the problem of setting $w_j$'s to minimize $\alpha$ and satisfy the gadget definition boils down to solving a large linear program which has inequalities dealing with every possible $a \in \{0,1\}^k$. This is almost obvious, except that the definition has elements of a 0-1 integer program: when $f(a) = 1$, we must ensure that there is an assignment $b$ making the sum equal $\alpha$. To circumvent this,~\cite{TSSW} also try {\em all possible functions} $B$ from the set of satisfying assignments of $f$ to $\{0,1\}^n$. Then, for every $a$ such that $f(a) = 1$, we simply set $b = B(a)$ in the ``$\sum_j w_j C_j(a,b) = \alpha$'' constraints of the linear program. As one might expect, this can lead to some very large linear programs, but for small constraint functions they are manageable.

Still, we had to fix $n$ to get a finite search space, and it is entirely possible that gadgets keep improving as $n$ increases. \cite{TSSW} prove that, for ${\cal F}$ satisfying very natural conditions, it suffices to set $n \leq 2^s$, where $s$ is the number of satisfying assignments of $f$. These conditions are satisfied by 2-CNF and many other well-studied constraint families.

Using the computer search for gadgets, the authors proved several interesting results in approximation which are still the best known to date; for example, $(16/17+\eps)$-approximating MAX CUT is {\sf NP}-hard. Incidentally, their search also uncovered an optimal 3.5-gadget reducing 3-SAT to MAX 2-SAT: taking $c_i$ as before, the 2-CNF gadget is \[(\ell_1 \vee \ell_3), (\neg \ell_1 \vee \neg \ell_3), (\ell_1 \vee \neg y_i), (\neg \ell_1 \vee y_i), (\ell_3 \vee \neg y), (\neg \ell_3 \vee y), (\ell_2 \vee y),\] where the weights are $1/2$ for every clause, except for the last one which has weight $1$. In the unweighted case, this amounts to having two copies of the last clause. The results here are feasibly checkable, for small constraints with a small number of auxiliary variables.

\subsubsection{Analyzing Approximation Algorithms}

Earlier, we noted that any 3-CNF formula can be approximated within $7/8$ by choosing a random assignment. But if some of the clauses are 2-CNF or 1-CNF, this no longer holds. However, it would be strange if we could not $7/8$-approximate general MAX 3-SAT because of this. Karloff and Zwick~\cite{KZ} proposed a possible $7/8$-approximation based on semidefinite programming (SDP).\footnote{For the purposes of this article, just think of semidefinite programming as a generalization of linear programming where the inequalities are between linear combinations of inner products of unknown vectors, and the task is to find vectors satisfying the inequalities. Such systems are approximately solvable in polynomial time.} Their algorithm is a fairly direct translation of MAX 3-SAT to an SDP, similar to the MAX CUT algorithm of Goemans and Williamson~\cite{GW}, where a vector $v_i$ in the solution corresponds to the variable $x_i$ in the formula, and one vector $v_t$ corresponds to TRUE. Given the vectors returned by the SDP solver, one obtains an assignment to the formula by picking a random hyperplane that passes through the origin and setting $x_i$ to TRUE if and only if $v_i$ and $v_t$ lie on different sides of the hyperplane. Such a hyperplane can be chosen by choosing a normal vector $r$ uniformly at random from the unit sphere in $\R^n$.

Analyzing the Karloff-Zwick algorithm is very difficult. In order to prove that the algorithm is a $7/8$-approximation, one needs to prove sharp bounds on the probability that four vectors $v_i,v_j,v_k,v_t$ from the SDP lie on the same side of a random hyperplane (i.e., the probability that the clause $(x_i \vee x_j \vee x_k)$ is falsified by the random assignment). This amounts to proving bounds on the {\em volume} of certain objects whose corners are chosen uniformly at random from the unit sphere, which Karloff and Zwick call ``volume inequalities.''

Karloff and Zwick were unable to prove strong enough inequalities to get a $7/8$-approximation, but they did obtain some partial results and gave a conjectured inequality that, if true, would imply that the algorithm is a  $7/8$-approximation. Zwick~\cite{Zwick} proved this inequality along with others by writing a program that used {\em interval arithmetic}. Interval arithmetic is a method for computing over real numbers on a computer in a controlled way, so that all errors are accounted for. The proofs of the Kepler conjecture and double-bubble conjecture mentioned earlier also utilize interval arithmetic in a critical way.

In interval arithmetic, one represents a real number $r$ by an interval $[r_0,r_1]$ where $r_0,r_1$ are machine representable and $r_0 \leq r \leq r_1$. Ideally, one wants $r_0$ ($r_1$) to be as large (small) as possible. For a real number $r$, let $\overline{r} = r_1$ and $\underline{r} = r_0$. We define \[[r_0,r_1] + [s_0,s_1] = [\underline{r_0+s_0},\overline{r_1+s_1}],\]\[[r_0,r_1]\cdot[s_0,s_1] = \left[\min\{\underline{r_0 s_0},\underline{r_0,s_1},\underline{r_1,s_0},\underline{r_1,s_1}\}, \max\{\overline{r_0 s_0},\overline{r_0,s_1},\overline{r_1,s_0},\overline{r_1,s_1}\}\right].\] One can define more complicated functions similarly. The point is that by doing numerical computations in interval arithmetic, the resulting interval {\em must} contain the correct value, even if that value cannot be machine represented.

But how can we use interval arithmetic to prove an inequality? The key step in Zwick's work is a technical reduction from the desired volume inequality to the task of proving that a certain system of constraints has {\em no} solution over the reals. (Most of these constraints are inequalities, but some are disjunctions of inequalities.) Zwick then wrote a program, called {\sc RealSearch}, which takes any system of constraints over bound variables of the form \[f_1(x_1,\ldots,x_n) \geq 0, \ldots, f_k(x_1,\ldots,x_n)\geq 0 ~\vee~ f_{k+1}(x_1,\ldots,x_n) \geq 0, \ldots,\] and tries to prove that they have no solution. Let the bounds on $x_i$ be $a_i \leq x_i \leq b_i$ for some machine representable $a_i,b_i$. The program starts by letting the interval $X_i = [a_i,b_i]$ to denote $x_i$, and evaluates the $f$-functions with interval arithmetic. If any of the constraints fail on this assignment (e.g., $f_1(X_1,\ldots,X_n) < 0$), then the system returns {\em no solution}. Otherwise, the program breaks some $X_i$ into subintervals $X'_i = [a_i,\overline{(a_i+b_i)/2}]$, $X''_i = [\underline{(a_i+b_i)/2},b_i]$, and recursively tries to verify that the system of constraints fails with both $X_i := X'_i$ and $X_i := X''_i$. Of course, such a procedure may not terminate, so the program is instructed to quit after some time. But surprisingly, this simple program can verify the necessary inequality, as well as several others that arise in SDP approximations! Even though we are not working explicitly over the reals, if we find that the system of constraints fails over all appropriately chosen intervals, then it follows that the system fails over all reals. In principle, this strategy could work for any functions $f$ definable in interval arithmetic.

Of course, the resulting proof of the volume inequality is of the infeasibly checkable variety, relying on the correctness of the program and the correctness of the floating-point operations. Even greater issues arise with Tom Hales' proof of the Kepler conjecture, which requires that his programs be run on a processor that strictly conforms to the IEEE-754 floating point standard. However, Zwick's strategy does not require that much stringency, and I believe it should be better known as a general method for attacking difficult inequalities.

\section{Time Lower Bounds}

I have recently found a nice domain in complexity theory where computer searches help perform the ``hard work'' in the proofs of theorems: namely, in proving time lower bounds for hard problems such as SAT on restricted computational models. In this case, the computer generates feasibly checkable proofs of lower bounds. Since this style of time lower bounds has been surveyed thoroughly by Van Melkebeek~\cite{VMsurvey,VMsurvey2}, I will not provide substantial background here. Instead I will focus more on describing how the reduction to a computer search works. For more details, please consult the draft available~\cite{fullversion}.

All the lower bounds amenable to computer search have one unifying property: the restricted model in which a lower bound is proved can be simulated {\em asymptotically faster} on an alternating machine. We call such a phenomenon a {\em speed-up property}. This property is crucial for the arguments to work. Here we will work with time lower bounds for SAT on random access machines that use only $n^{o(1)}$ workspace. Define $\DTS[t(n)]$ to be the class of problems solvable by such machines in $t(n) \geq n$ time (the acronym stands for ``deterministic time with small space''). Here is one example of the speed-up property in this setting.

\begin{theorem}\label{1stspeed} $\DTS[t(n)] \isin \Sig[2] \T[t(n)^{1/2} n^{o(1)}] \cap \Pie[2] \T[t(n)^{1/2} n^{o(1)}]$.
\end{theorem}

(The classes $\Sig[2] \T$ and $\Pie[2] \T$ are defined in the usual way.) That is, we can simulate a small space computation with a square-root speedup using alternations. The proof is due to Kannan~\cite{Kannan}, but the basic idea goes back to Savitch~\cite{Savitch}. The idea is to guess snapshots of the $n^{o(1)}$ space algorithm at $t^{1/2}$ points during its computation, then verify in parallel that the guesses are correct. Given an algorithm $A$ that runs in time $t$ and uses space $n^{o(1)}$, the corresponding $\Sig[2] $ algorithm $B(x)$ {\em existentially} writes $t^{1/2}$ configurations $C_0,\ldots,C_{t^{1/2}}$ of $A(x)$, where $C_0$ is the initial configuration of $A(x)$ and $C_{t^{1/2}}$ is an accepting configuration. Since $A$ uses only $n^{o(1)}$ space, these configurations can be written down with $n^{o(1)}$ bits each. Next, $B(x)$ {\em universally} writes $i \in \{0,\ldots,t^{1/2}-1\}$ and jumps to the configuration $C_i$. Then it simulates $A(x)$ from $C_i$ for $t^{1/2}$ steps, accepting if and only if $A(x)$ is in configuration $C_{i+1}$. The $\Pie[2] $-simulation can be defined analogously.

Theorem~\ref{1stspeed} is already enough to prove a non-trivial lower bound for SAT, after applying a few more observations from the literature. The first observation is that if SAT is in $\DTS[n^c]$, then $\NTIME[n] \isin \DTS[n^c \cdot \poly(\log n)]$. This follows from the fact that SAT is very strongly {\sf NP}-complete:

\begin{theorem}[\cite{Cook,Schnorr,FLvMV}]\label{SATred} For every $L \in \NTIME[n]$, there is a reduction from $L$ to SAT that maps strings of length $n$ to formulas of size $n \poly(\log n)$, where an arbitrary bit of the reduction can be computed in $\poly(\log n)$ time.
\end{theorem}

The proof is a very technical version of Cook's theorem which we will not describe here, but let us note in passing that other problems such as Vertex Cover also enjoy a similar property. By padding and the fact that $\DTS$ classes are closed under complement, we have the following.

\begin{theorem}\label{1stslow} If SAT is in $\DTS[n^c]$, then for all $k$ and $t(n) \geq n$, $\Sig[k] \T[t(n)] \isin \Sig[k-1] \T[t(n^c)]$ and $\Pie[k] \T[t(n)] \isin \Pie[k-1] \T[t(n^c)]$.
\end{theorem}

Theorem~\ref{1stslow} says we can remove alternations from a computation, with a small slowdown in runtime. (For this reason, I like to call it a ``slow-down theorem.'') Theorem~\ref{1stspeed} says we can add alternations to a $\DTS$ computation, with a speedup in running time. Naturally, one's inclination is to pit these two results against one another and see what we can derive. Assuming SAT is in $\DTS[n^c]$, we find \[ \NT[n^2] \isin \DTS[n^{2c}] \isin \Sig[2] \T[n^{c}] \isin \NTIME[n^{c^2}],\] where the first and third containments follow from Theorem~\ref{1stslow}, and the second containment follows from Theorem~\ref{1stspeed}. When $c < 2^{1/2}$, the above contradicts the nondeterministic time hierarchy~\cite{Cook72}. We have proved the following:

\begin{theorem}[\cite{FLvMV}] SAT cannot be solved by an algorithm that runs in $n^{\sqrt{2}-\eps}$ time and $n^{o(1)}$ space, for every $\eps > 0$.\end{theorem}

With a more complicated argument involving the same tools, \cite{FLvMV} proved that SAT cannot be in $\DTS[n^{\phi-\eps}]$, where $\phi = 1.618\ldots$ is the golden ratio. We can prove a simple $n^{1.6}$ lower bound by generalizing Theorem~\ref{1stspeed}.

At this point it will be helpful to introduce some new notation, and this notational shift is crucial for the automated approach. Letting $t(n)$ be a polynomial and letting $b \geq (\log t(n))/(\log n)$, define the class $(\exists~t(n))^{b}{\cal C}$ to be the class of problems solvable by a machine that existentially guesses $t(n)$ bits, then selects $O(n^b)$ of those bits (along with the input) and feeds them as input to a representative machine from class ${\cal C}$. (The selection procedure is required to take only linear time and logarithmic space, so it does not interfere with any of the time/space constraints of the class.) We define $(\forall~t(n))^b {\cal C}$ similarly. By properties of nondeterminism and co-nondeterminism, note that we can ``combine'' adjacent quantifiers in a class:

\begin{proposition}\label{combine} $(\exists~t_1(n))^{b_1}(\exists~t_2(n))^{b_2}{\cal C} = (\exists~t_1(n)+t_2(n))^{b_2}{\cal C}$, and the analogous statement with $\forall$ also holds.
\end{proposition}

Theorem~\ref{1stspeed} can now be stated more generally:

\begin{theorem}[Speedup Rule]\label{2ndspeed} For all $x$ such that $n \leq n^x \leq t(n)$, \[\DTS[t(n)] \isin (\exists~n^{x+o(1)})^{x}(\forall~\log x)^{1}\DTS[t(n)/n^x].\] The theorem also holds when we interchange $\forall$ and $\exists$.
\end{theorem}

Theorem~\ref{2ndspeed} holds because we can just guess $n^b+1$ configurations (instead of $t(n)^{1/2}+1$ as before), universally pick $i \in \{0,\ldots,n^b\}$, and the input to the final $\DTS$ computation will simply be the original input along with the pair of configurations $(C_i,C_{i+1})$, which has size $n^{o(1)}$. Hence we have $n^{b+o(1)}$ in the $\exists$-quantifier, and $O(n)$ bits of input to the final $\DTS$ class.

Our new notation also lets us to state the ``slow-down theorem'' in a more precise way.

\begin{theorem}[Slowdown Rule]\label{2ndslow} If SAT is in $\DTS[n^c]$ then for all $a_1,b_1$, $\ldots$, $a_k,b_k$, $a_{k+1} \geq 1$, and $Q_i \in \{\exists,\forall\}$, the class $(Q_1~n^{a_1})^{b_1}\cdots (Q_{k-1}~n^{a_{k-1}})^{b_{k-1}}(Q_k~n^{a_k})^{b_k}\DTS[n^{a_{k+1}}]$ is contained in the class $(Q_1~n^{a_1})^{b_1}\cdots (Q_{k-1}~n^{a_{k-1}})^{b_{k-1}}\DTS[n^{c \cdot \max\{b_{k-1},a_{k},a_{k+1}\}+o(1)}]$.
\end{theorem}

Again, the result holds by Theorem~\ref{SATred} and a standard padding argument. In particular, $^{b_{k-1}}(Q_k~n^{a_k})^{b_k}\DTS[n^{a_{k+1}}]$ is contained in either $\NTIME[n^{\max\{b_{k-1},a_{k},a_{k+1}\}}]$ or $\coNTIME[n^{\max\{b_{k-1},a_{k},a_{k+1}\}}]$, and both of these are in $\DTS[n^{c \cdot \max\{b_{k-1},a_{k},a_{k+1}\}+o(1)}]$. We are now ready to prove a stronger time lower bound for SAT.

\begin{theorem}\label{1.6} SAT cannot be solved by an algorithm running in $n^{1.6}$ time and $n^{o(1)}$ space.\end{theorem}

\begin{proof} Suppose SAT $\in \DTS[n^c]$ where $\sqrt{2} \leq c \leq 1.6$. By Theorem~\ref{SATred}, $\NTIME[n] \isin \DTS[n^{c+o(1)}]$. We can derive $\NT[n^{c/2 + 2/c}] \isin \NTIME[n^{c^3/2+o(1)}]$, ignoring $o(1)$ factors for simplicity: \begin{eqnarray*}\NTIME[n^{c/2+2/c}] &\isin& \DTS[n^{c^2/2 + 2}]~~~\text{(Slowdown)}\\
&\isin &(\exists~n^{c^2/2})^{c^2/2}(\forall~\log n)^1\DTS[n^2]~~~\text{(Speedup, with $x = c^2/2$)}\\
& \isin & (\exists~n^{c^2/2})^{c^2/2}(\forall~\log n)^1(\forall~n)^1(\exists~\log n)^1\DTS[n]~~~\text{(Speedup, with $x = 1$)}\\
& = & (\exists~n^{c^2/2})^{c^2/2}(\forall~n)^1(\exists~\log n)^1\DTS[n]~~~\text{(Proposition~\ref{combine})}\\
& \isin & (\exists~n^{c^2/2})^{c^2/2}(\forall~n)^1\DTS[n^c]~~~\text{(Slowdown)}\\
& \isin & (\exists~n^{c^2/2})^{c^2/2}\DTS[n^{c^2}]~~~\text{(Slowdown)}\\
& \isin & (\exists~n^{c^2/2})(\exists~n^{c^2/2})^{c^2/2}(\forall~\log n)^{c^2/2}\DTS[n^{c^2/2}]~~~\text{(Speedup, $x = c^2/2$)}\\
& = & (\exists~n^{c^2/2})^{c^2/2}(\forall~\log n)^{c^2/2}\DTS[n^{c^2/2}]~~~\text{(Proposition~\ref{combine})}\\
& \isin & (\exists~n^{c^2/2})^{c^2/2}\DTS[n^{c^3/2}]~~~\text{(Slowdown)}\\
& \isin & \NTIME[n^{c^3/2}].\end{eqnarray*}

When $c/2 + 2/c > c^3/2$ (which happens for $c < \sqrt{\frac{1+\sqrt{17}}{2}} \approx 1.6004$) we have a contradiction to the nondeterministic time hierarchy. \end{proof}

The best known SAT time lower bound (with $n^{o(1)}$ space) is from an earlier paper of ours~\cite{Will}, and it achieves $n^{2 \cos(\pi/7)} \geq n^{1.8}$. It is a fairly elaborate inductive proof that builds on the same ideas. A natural question is if we can do even better than this. It is clear that we have a very specific type of proof system on our hands; it is also powerful, in that all known time-space lower bounds for SAT (and QBF) on random access machines work over it.

There is nothing too complicated about the proof of Theorem~\ref{1.6}: we are just applying the Speedup and Slowdown Rules in clever ways. More interestingly, the proof was discovered by a computer program. Furthermore it is the best lower bound one can prove with only $7$ applications of Speedup and Slowdown Rules, and we know this because a computer program tried all the cases.

How did it try the cases? It would seem that the space of possibilities is too high: how could a computer try all possible expressions for the exponents? One can show that {\em once we have specified the sequence in which the Speedup and Slowdown Rules are applied, the task of finding the optimal lower bound argument can be formulated as a linear program}. This makes our job of finding good lower bound proofs much easier.

Let me sketch how a linear program can be constructed, for a fixed sequence of rules to apply. One can show that the sequence of rules completely determines the number of quantifiers in each class in the chain of inclusions of the proof. (There are actually two ways to apply the Speedup Rule, one where we introduce a $\Sig[2] $ computation and the other a $\Pie[2] $ computation, but we can prove that one of these applications is always superior.) Suppose we have a sequence of inclusions such as those in the proof of Theorem~\ref{1.6}, but all exponents in the polynomials are replaced with variables. So for example, the second inclusion (or ``line'') in the proof of Theorem~\ref{1.6} becomes \[(\exists~n^{a_{2,3}})^{b_{2,2}}(\forall~n^{a_{2,2}})^{b_{2,1}}\DTS[n^{a_{2,1}}],\] the initial class $\NTIME[n^{c/2+2/c}]$ is replaced with $\NT[n^{a_{0,1}}]$, and the final class $\NTIME[n^{c^3/2}]$ becomes $\NT[n^{a_{8,1}}]$. In general, we replace the exponents in the $i$th line with variables $a_{i,j}$, $b_{i,j}$. Now we want to write a linear program in terms of these variables that expresses the applications of the two rules and captures the fact that we want a contradiction. To do the latter is very easy: we simply require $a_{8,1} < a_{0,1}$, or $a_{8,1} \leq a_{0,1}-\eps$ for some $\eps > 0$. To express a Speedup Rule on the $i$th line, we introduce a parameter $x_i \geq 0$ and include the following inequalities: \[\begin{array}{c} a_{i,1} \geq 1,~ a_{i,1} \geq a_{i-1,1} - x_i,~b_{i,1} = b_{i-1,1},~a_{i,2} = 1,~b_{i,2} \geq x_i,~b_{i,2} \geq b_{i-1,1},~ a_{i,3} \geq a_{i-1,2},\\~a_{i,3} \geq x_i,~(\forall ~k:~4 \leq k \leq m)~a_{i,k} = a_{i-1,k-1}, (\forall ~k:~4 \leq k \leq m-1)~b_{i,k} = b_{i-1,k-1}.\end{array}\] \indent Intuitively, these constraints express that the class $(Q_m~n^{a_m})^{b_{m-1}}\cdots^{b_2}(Q_2~ n^{a_2})^{b_1}\DTS[n^{a_1}]$ on the $(i-1)$th line is replaced with $(Q_m~n^{a_m})^{b_{m-1}}\cdots^{b_2}(Q_2~ n^{\max\{a_2,x_i\}})^{\max\{x_i,b_1\}}(Q_1~n)^{b_1}\DTS[n^{\max\{a_1-x_i,1\}}]$ on the $i$th line, where $Q_1$ is the quantifier opposite to $Q_2$. (One can check that this indeed simulates the Speedup Rule faithfully.) We can express the Slowdown Rule in a similar way, by treating the desired lower bound exponent $c$ as a constant. Note if we minimize the sum of $a_{i,j}+b_{i,j}$ over all $i,j$, then the above inequalities faithfully simulate $\max$.

Given that we can take a sequence of rules and turn it into an LP that can then be solved, what sequences are good to try? The number of sequences to search can be reduced by establishing several properties of the proof system that hold without loss of generality. For example, we may always start with $\DTS[n^k]$ for some $k$, and if we derive $\DTS[n^k] \subseteq \DTS[n^{k-\eps}]$ then we have a contradiction, and every proof that works otherwise can be rewritten to work like this. There are several simplifications of this type and while their proofs are not very enlightening, as a whole they let us identify the relevant parts of proofs. They also help us prove {\em limitations} on the proof system.

The chart on the next page gives a graph of experimental results from a search for short proofs of time-space lower bounds for SAT. The $x$-coordinate is the number of lines in a proof (the number of Speedup/Slowdown applications) and the $y$-coordinate is the exponent of the best lower bound attained with that number of lines.

\begin{figure}

\input{SAT2.tex}

\end{figure}

Up to 25 lines, the search was completely exhaustive. Beyond that, I used a heuristic search that takes a queue of best-found proofs for small lengths and tries to locally improve them by inserting new rule applications. When a better lower bound is found, the new proof is added to the queue. This heuristic search was run up to about $50+$ lines. Now, all of the best proofs found up to $50+$ lines have a certain pattern to them, resembling the structure of the $2 \cos(\pi/7)$ lower bound. Restricting the search to work only within this pattern, we can get annotations for $70+$ lines which still exhibit the pattern. Checking a $383$-line proof of similar form, the lower bound attained was $n^{1.8017}$, very close to $2 \cos(\pi/7) \approx 1.8019$. These experimental results lead us to:

\begin{conjecture} The best time lower bound for SAT (in $n^{o(1)}$ space) that can be proved with the above proof system is the $n^{2 \cos(\pi/7)}$ bound of~\cite{Will}.
\end{conjecture}

Given the scale at which the conjecture has been verified, I am fairly confident in its truth although I do not know how to prove it. Unlike~\cite{TSSW}, we do not know how to place a finite upper bound on all the parameters, namely the lengths of proofs. The conjecture is indeed surprising, if true. The general sentiment among researchers I have talked to (and anonymous referees from the past) was that a quadratic time lower bound (or more precisely, $n^{2-\eps}$ for all $\eps > 0$) should be possible with the ingredients we already have. We can show formally that a better lower bound than this cannot be established with the current approach.

\begin{theorem} In the above proof system, one cannot prove that SAT requires $n^2$ time with $n^{o(1)}$ space algorithms.
\end{theorem}

This theorem is proven by minimal counterexample: we take a minimum proof of a quadratic lower bound, and show that there is a subsequence of rules that can be removed such that the underlying LP remains feasible with the same parameters as before. So one possible strategy for proving the conjecture is to find some subsequence of rules that must arise in any optimal proof, and show that if one assumes $c > 2 \cos(\pi/7)$ then this subsequence can be removed from the proof without weakening it. However this strategy appears to be difficult to carry out.

\section{Finding Small Circuits}

In this last section, I will speculate about a computational approach to understanding Boolean circuit complexity. The following has been a joint effort with Maverick Woo.

Our knowledge of Boolean circuit complexity is quite poor. (For concreteness, let us concentrate on circuits comprised of AND, OR, and NOT gates.) We do not know how to prove strong circuit lower bounds for problems in {\sf P}; the best known is $5n$~\cite{LR01,IM02}. One good reason why we don't know much about the true power of circuits is that we don't have many examples of minimum circuits. We don't know, for example, what an optimal circuit for $3 \times 3$ Boolean matrix multiplication looks like.

It is possible that we could make progress in understanding circuits by cataloging the smallest circuits we know for basic functions, on small input sizes (such as $n=1,\ldots,10$). This suggestion makes more sense for some problems than others. For SAT, the circuit complexity can depend on the encoding of Boolean formulas; for matrix operations, the encoding is clear. Sloane and Plouffe first published~\cite{SP} and now maintain the Online Encyclopedia of Integer Sequences, an exhaustive catalog of interesting sequences that arise in mathematics and the sciences. Might we benefit from an Encyclopedia of Minimum Circuits? For example, what do the smallest Boolean circuits for 10 $\times$ 10 Boolean matrix multiplication look like? Are they regular in structure?  It is likely that the answers would give valuable insight into the complexity of the problem. The best algorithms we know of reduce the problem to matrix multiplication over a ring, which is then solved by a highly regular, recursive construction (such as Strassen's~\cite{Strassen}). Even if the cataloged circuits are not truly minimal but close to that, concrete examples for small inputs could be useful for theoreticians to mine for inspiration, or perhaps for computers to mine for patterns via machine learning techniques. The power of small examples cannot be underestimated.

How can we get small examples of minimum circuits? One potential approach is to reduce this task to the task of developing good solvers for quantified Boolean formulas, an area in AI that has seen much technical progress lately. We can pose the problem of finding a small circuit as a quantified Boolean formula (QBF), and feed the QBF to one of many recently developed QBF solvers. A QBF $\Phi_{s,n}$ for size-$s$ $n \times n$ matrix multiplication circuits can be stated roughly as: \[\Phi_{s,n} = (\exists~\text{circuit}~C~\text{of}~s~\text{gates}, 2n^2~\text{inputs}, n^2~\text{outputs})(\forall~n\times n~\text{matrices}~X,Y)[X\cdot Y = C(X,Y)],\] where the predicate can be easily encoded as a SAT instance. In our encoding, we allow the circuits to have unbounded fan-in. For simplicity, we searched for circuits made up of only {\tt NOR} gates.

Experiments with QBF solvers have not yet revealed significant new insight. So far, they have discovered one fact: {\em the optimal size circuit for $2 \times 2$ Boolean matrix multiplication is the obvious one}. Well, duh. What about the $3 \times 3$ case? This is already difficult! The {\tt sKizzo} QBF solver~\cite{skizzo} can prove that there is no circuit for $3 \times 3$ that has 10 gates, but nothing beyond that. Even when we restrict the gates to have fan-in two, the solver crashes on larger instances.

I do not see this limited progress as a substantial deterrent. On the practical side, QBF solvers have only seen serious scientific attention in the last several years, and huge developmental strides have already been made. On the theoretical side, if we look at matrix multiplication over a field, a better way to approach the problem would be to phrase the QBF for small circuits as something like a ``Merlin-Arthur formula'': that is, we guess the circuit to be used and verify that it computes the product by evaluating on random matrices. In that case, the instances should be much easier to solve. In general, we can find approximately minimum circuits for problems with polysize circuits in ${\sf ZPP}^{\sf NP}$~\cite{Bshouty}. By using a high-quality SAT solver in place of the ${\sf NP}$ oracle, the idea of building an approximate circuit encyclopedia does not seem too implausible. However some effort will be needed to adapt the results to work in practice.

I do believe that in the near future, the general problem of finding small minimal circuits for {\sf P} problems will be within the reach of practice. Analyzing these new gadgets should inject a fresh dose of ideas in the area of circuit complexity.

%\section{When are automated proofs appropriate?}

%From our experience and knowledge of prior work, the use of computer searches in proving theorems typically follows after a deep insight into the structures pertaining to the theorem statement. This insight leads one to find a suitable formalization of the problem so that a somewhat tractable brute-force search remains. That is, the computer search itself does not appear to yield much further insight into the problem.

%When *you* (=you) or *I* (=Lane) run latex, the page
%  footers will say something crazy about the issue
%  month/year/volume/number.  But don't worry---when the
%  editor-in-chief runs it, he will have in the right place the right
%  magic file that the style file uses to set these, and so for him it
%  will in theory come out right.


\begin{thebibliography}{99}

\bibitem[ABBIMP05]{Borodin} M. Alekhnovich, A. Borodin, J. Buresh-Oppenheim, R. Impagliazzo, A. Magen, and T. Pitassi. Toward a model for backtracking and dynamic programming. {\em Proc. IEEE Conference on Computational Complexity}, 308--322, 2005.

\bibitem[AH77a]{AH1} K. Appel and W. Haken. Every planar map is four colorable. Part I. Discharging. {\em Illinois J. Math.} 21:429--490, 1977.

\bibitem[AH77b]{AH2} K. Appel, W. Haken, and J. Koch. Every planar map is four colorable. Part II. Reducibility. {\em Illinois J. Math.} 21:491--567, 1977.

\bibitem[AS98]{AS} S. Arora and S. Safra. Probabilistic checking of proofs: A new characterization of NP. {\em J.~ACM} 45(1):70--122, 1998.

\bibitem[ALMSS98]{ALMSS} S. Arora, C. Lund, R. Motwani, M. Sudan and M. Szegedy. Proof verification and the hardness of approximation problems. {\em J.~ACM} 45(3):501--555, 1998.

\bibitem[BGS98]{BGS} M. Bellare, O. Goldreich, and M. Sudan. Free bits, PCPs, and non-approximability: towards tight results. {\em SIAM J.~Comput.} 27(3):804--915, 1998.

\bibitem[Ben05]{skizzo} M. Benedetti. sKizzo: A suite to evaluate and certify QBFs. {\em Proc. Int'l Conf. on Automated Deduction}, 369--376, 2005.

\bibitem[BCGKT96]{Bshouty} N. H. Bshouty, R. Cleve, R. Gavald\'{a}, S. Kannan, and C. Tamon. Oracles and queries that are sufficient for exact learning. {\em J.~Comput. Syst. Sci.} 52(3):421--433, 1996.

\bibitem[Coo72]{Cook72} S. A. Cook. A hierarchy for nondeterministic time complexity. {\em Proc. ACM STOC}, 187--192, 1972.

\bibitem[Coo88]{Cook} S. A. Cook. Short propositional formulas represent nondeterministic computations. {\em Information Processing Letters} 26(5): 269-270, 1988.

\bibitem[Epp06]{Eppstein} D. Eppstein. Quasiconvex analysis of multivariate recurrence equations for backtracking algorithms. {\em ACM Trans. on Algorithms} 2(4):492--509, 2006.

\bibitem[FK06]{FK06} S. S. Fedin and A. S. Kulikov. Automated proofs of upper bounds on the running time of splitting algorithms. {\em J.~Math. Sciences} 134(5):2383--2391, 2006.

\bibitem[FGK05a]{FGK1} F. V. Fomin, F. Grandoni, and D. Kratsch. Measure and conquer: domination - a case study. {\em Proc. ICALP}, 191--203, 2005.

\bibitem[FGK05b]{FGK3} F. V. Fomin, F. Grandoni, and D. Kratsch. Some new techniques in design and analysis of exact (exponential) algorithms. {\em Bulletin of the EATCS} 87:47--77, 2005.

\bibitem[FGK06]{FGK2} F. V. Fomin, F. Grandoni, and D. Kratsch. Measure and conquer: a simple $O(2^{0.288n})$ independent set algorithm. {\em Proc. ACM-SIAM SODA}, 18--25, 2006.

\bibitem[FLvMV05]{FLvMV} L. Fortnow, R. Lipton, D. van Melkebeek, and A. Viglas. Time-space lower bounds for satisfiability. {\em Journal of the ACM} 52(6):835--865, 2005.

\bibitem[GJS76]{GJS} M. Garey, D. Johnson, and L. Stockmeyer. Some simplified NP-complete graph problems. {\em Theor. Comput. Sci.} 1:237--267, 1976.

\bibitem[GW95]{GW} M. Goemans and D. Williamson. Improved approximation algorithms for maximum cut and satisfiability problems using semidefinite programming. {\em J.~ACM} 42:1115--1145, 1995.

\bibitem[GGHN04]{Gramm} J. Gramm, J. Guo, F. H\"{u}ffner, and R. Niedermeier. Automated generation of search tree algorithms for hard graph modification problems. {\em Algorithmica} 39:321--347, 2004.

\bibitem[Hal05]{Hales} T. C. Hales. A proof of the Kepler conjecture. {\em Annals of Math.} 162:1065--1185, 2005.

\bibitem[HS00]{HS} J. Hass and R. Schlafly. Double bubbles minimize. {\em Annals of Math.} 151:459--515, 2000.

\bibitem[Has01]{Hastad} J. H\aa stad. Some optimal inapproximability results. {\em J.~ACM} 48:798--859, 2001.

\bibitem[IM02]{IM02} K. Iwama and H. Morizumi. An explicit lower bound of $5n - o(n)$ for Boolean circuits. {\em Proc. MFCS}, 353--364, 2002.

\bibitem[IT04]{IT} K. Iwama and S. Tamaki. Improved upper bounds for 3-SAT. {\em Proc. ACM SODA'04}, 321--322, 2004.

\bibitem[Kan84]{Kannan} R. Kannan. Towards separating nondeterminism from determinism. {\em Mathematical Systems Theory} 17(1):29--45, 1984.

\bibitem[KZ97]{KZ} H. J. Karloff and U. Zwick. A 7/8-approximation algorithm for MAX 3SAT? {\em Proc. IEEE FOCS} 406--415, 1997.

\bibitem[KK06]{KK06} A. Kojevnikov and A. S. Kulikov. A new approach to proving upper bounds for MAX-2-SAT. {\em Proc. ACM-SIAM SODA}, 11--17, 2006.

\bibitem[LR01]{LR01} O. Lachish and R. Raz. Explicit lower bound of $4.5n - o(n)$ for Boolean circuits. {\em Proc. ACM STOC}, 399--408, 2001.

\bibitem[LTS89]{Lam} C. W. H. Lam, L. Thiel, and S. Swiercz. The nonexistence of finite projective planes of order 10. {\em Canad. J.~Math.} 41:1117--1123, 1989.

\bibitem[McC97]{Robbins} W. McCune. Solution of the Robbins problem. JAR 19(3):263--276, 1997.

\bibitem[vM04]{VMsurvey} D. van Melkebeek. Time-space lower bounds for NP-complete problems. In {\em Current Trends in Theoretical Computer Science} 265--291, World Scientific, 2004.

\bibitem[vM07]{VMsurvey2} D. van Melkebeek. A survey of lower bounds for satisfiability and related problems. {\em Foundations and Trends in TCS} 2:197--303, 2007.

\bibitem[NS03]{NS} S. I. Nikolenko and A. V. Sirotkin. Worst-case upper bounds for SAT: automated proof. {\em Proc. ESSLLI} 225--232, 2003. URL: {\tt http://logic.pdmi.ras.ru/\~{}sergey/}

\bibitem[RSST97]{RSST} N. Robertson, D. P. Sanders, P. D. Seymour and R. Thomas. A new proof of the four colour theorem. {\em J.~Combinatorial Theory B} 70:2--4, 1997.

\bibitem[Rob01]{Robson} M. Robson. Finding a maximum independent set in time $O(2^{n/4})$. Technical Report 1251-01, LaBRI, Universit\'{e} de Bordeaux I, 2001. \\ URL: {\tt http://www.labri.fr/perso/robson/mis/techrep.html}

\bibitem[Sav70]{Savitch} W. J. Savitch. Relationships between nondeterministic and deterministic tape classes. {\em J.~Comp. Sys. Sci.} 4:177--192, 1970.

\bibitem[Sch78]{Schnorr} C. Schnorr. Satisfiability is quasilinear complete in NQL. {\em Journal of the ACM} 25(1):136--145, 1978.

\bibitem[SS07]{SS} A. D. Scott and G. B. Sorkin. Linear-programming design and analysis of fast algorithms for Max 2-CSP. {\em Discr. Optimization} 4(3-4):260--287, 2007.

\bibitem[SP95]{SP} N.~J.~A. Sloane and S. Plouffe. The encyclopedia of integer sequences. Academic Press, 1995.

\bibitem[Str69]{Strassen} V. Strassen. Gaussian elimination is not optimal. {\em Numer. Math.} 13:354--356, 1969.

\bibitem[TSSW00]{TSSW} L. Trevisan, G. Sorkin, M. Sudan, D. P. Williamson. Gadgets, approximation, and linear programming. {\em SIAM J.~Computing} 29(6):2074--2097, 2000.

\bibitem[Wil07]{Will} R. Williams. Time-space tradeoffs for counting NP solutions modulo integers. {\em Computational Complexity} 17(2):179--219, 2008.

\bibitem[Wil08]{fullversion} R. Williams. Automated proofs of time lower bounds. Manuscript available at {\tt http://www.cs.cmu.edu/\~{}ryanw/projects.html}.

\bibitem[Zwi02]{Zwick} U. Zwick. Computer assisted proof of optimal approximability results. {\em Proc. ACM-SIAM SODA} 496--505, 2002.

\end{thebibliography}
\end{document}